%% file: main.tex
\documentclass{article}
\usepackage[margin=1in]{geometry}
\input{preamble}

\title{Online Allocation with Concave, Diminishing-Returns Objectives}
\author{Kalen Patton\thanks{
        (kpatton33@gatech.edu)
        School of Mathematics,
        Georgia Tech.
        Supported in part by NSF awards CCF-2327010 and CCF-2440113.
        }}

\begin{document}
\maketitle
\begin{abstract}
Online resource allocation problems are central challenges in economics and computer science, modeling situations in which $n$ items arriving one at a time must each be immediately allocated among $m$ agents. In such problems, our objective is to maximize a monotone reward function $f(\mathbf{x})$ over the allocation vector $\mathbf{x} = (x_{ij})_{i, j}$, which describes the amount of each item given to each agent. In settings where $f$ is concave and has ``diminishing returns'' (monotone decreasing gradient), several lines of work over the past two decades have had great success designing constant-competitive algorithms, including the foundational work of Mehta et al. (2005) on the Adwords problem and many follow-ups. Notably, while a greedy algorithm is $\frac{1}{2}$-competitive in such settings, these works have shown that one can often obtain a competitive ratio of $1-\frac{1}{e} \approx 0.632$ in a variety of settings when items are divisible (i.e. allowing fractional allocations). However, prior works have thus far used a variety of problem-specific techniques, leaving open the general question: \emph{Does a $(1-\frac{1}{e})$-competitive fractional algorithm always exist for online resource allocation problems with concave, diminishing-returns objectives?}

In this work, we answer this question affirmatively, thereby unifying and generalizing prior results for special cases. Our algorithm is one which makes continuous greedy allocations with respect to an auxiliary objective $U(x)$. Using the online primal-dual method, we show that if $U$ satisfies a ``balanced'' property with respect to $f$, then one can bound the competitiveness of such an algorithm. Our crucial observation is that there is a simple expression for $U$ which has this balanced property for any $f$, yielding our general $(1-\frac{1}{e})$-competitive algorithm.

\end{abstract}

\newpage
\input{intro}

\input{problem-definition}

\input{main-proof}

\appendix

\input{appendix}

\bibliographystyle{alpha}
\bibliography{ref}

\end{document}

%% file: preamble.tex
\usepackage[utf8]{inputenc}
\usepackage{parskip}

\usepackage{babel}
\usepackage{amsmath, amsthm}
\usepackage{amssymb, enumerate}
\usepackage{dsfont, xcolor}
\usepackage{hyperref}
\usepackage{cleveref}
\usepackage{algorithm}
\usepackage[noend]{algpseudocode}

\newcommand{\R}{\mathbb{R}}
\newcommand{\Rp}{\mathbb{R}_{\geq 0}}

\newcommand{\ind}{\mathds{1}}

\newcommand{\up}[1]{^{({#1})}}
\newcommand{\ip}[1]{\langle{#1}\rangle}

\newcommand{\bb}[1]{\mathbf{{#1}}}
\newcommand{\bu}{\mathbf{u}}
\newcommand{\bx}{\mathbf{x}}
\newcommand{\by}{\mathbf{y}}
\newcommand{\bz}{\mathbf{z}}

\newtheorem{theorem}{Theorem}[section]
\newtheorem{lemma}[theorem]{Lemma}

\newtheorem{proposition}[theorem]{Proposition}

\newtheorem{definition}[theorem]{Definition}

%% file: intro.tex
\section{Introduction}
Online resource allocation problems are central challenges in computer science and economics, and as a consequence, they have received considerable attention over the past few decades. For such problems, a sequence of items $j \in [m]$ (e.g. ads, goods, jobs, etc.) arrives online, and our algorithm must immediately allocate each item among a set of offline agents $[n]$. Each agent $i \in [n]$ has a valuation function $v_i$, and receives utility $u_i = v_i((x_{ij})_{j \in [m]})$, where $x_{ij}$ denotes the amount of item $j$ allocated to $i$. The goal of the algorithm is to maximize some aggregate function $W(u_1, \dots, u_n)$ over the vector of buyers' utilities (e.g. $\sum_i u_i$, Nash welfare $\prod_i u_i^{1/n}$, etc.). The performance of such an algorithm is measured by the \emph{competitive ratio}, i.e. the ratio of the algorithm's expected welfare to the hindsight optimum. 

In this work, we consider such problems in which the objective $f(\bx) := W(v_1(\bx_1), \dots, v_n(\bx_n))$ is (1) concave, and (2) has diminishing returns (i.e. monotone decreasing gradients) as a function of the allocation vector $\bx$. These are natural properties which arise frequently in resource allocation settings, such as in \cite{mehta2005adwords,devanur2012online,devanur2013whole,wang2016matroid,zhang2020combinatorial,hathcock2024online}. It is not difficult to show that any online problem with these properties admits a $\frac{1}{2}$-competitive greedy algorithm. Hence, the fundamental question in such settings is if one can obtain a competitive ratio $> \frac{1}{2}$.

One of the most famous problems in this framework is the Adwords problem studied by Mehta, Saberi, Vazirani, and Vazirani \cite{mehta2005adwords}. In the Adwords problem, ads must be allocated online to agents with budget-additive valuation functions $v_i(\bx_i) = \min\{\sum_{j} b_{ij}x_{ij},~B_i\}$ in order to maximize the sum of buyer utilities $W(u) = \sum_i u_i$. The key result of Mehta et al. is that one can obtain a competitive ratio $1-\frac{1}{e} \approx 0.632$ when buyer budgets $B_i$ are large compared to the bids $b_{ij}$, or when items are divisible (i.e. allowing fractional $x_{ij}$). In recent years, numerous works have extended this $(1-\frac{1}{e})$-competitive algorithm to a variety of fractional or ``small-bids'' settings. Some of these lines of work include the following:
\begin{enumerate}
    \item \textbf{Non-uniform Item Weights.} Feldman, Korula, Mirrokni, and P\'al \cite{feldman2009online} considered a weighted setting where items have reward values $w_{ij}$ distinct from the budget they consume $b_{ij}$. They show that one can still obtain a $(1 - \frac{1}{e})$-competitive algorithm with the addition of a ``free disposal'' assumption, which is equivalent to maximizing the sum of utilities $u_i = \max\{ \sum_j w_{ij} z_{ij} : \sum_j b_{ij} z_{ij} \leq B_i,~\bz \leq \bx\}$.
    \item \textbf{Concave Valuations.} Devanur and Jain \cite{devanur2012online} considered a setting in which agents have $u_i = M_i(\sum_j b_{ij} x_{ij})$ for a monotone concave $M_i : \Rp \to \Rp$. They showed that each such concave function $M$ admits an optimal competitive ratio $F(M)$, where one always has $F(M) \geq 1 - \frac{1}{e}$.
    \item \textbf{Simultaneous Rewards.} Devanur, Huang, Korula, and Mirrokni \cite{devanur2013whole} studied a variation of Adwords in which ``pages'' of ads arrive online, each with a set of possible configurations in which it can be allocated. Each configuration can give a reward to multiple agents simultaneously. The author show that, even when allocations give rewards to multiple agents, a version of the algorithm of \cite{mehta2005adwords} gives a $1-\frac{1}{e}$ competitive ratio.
    \item \textbf{Combinatorial Budgets.} A series of works \cite{wang2016matroid, zhang2020combinatorial, hathcock2024online} have considered the setting when agents' utilities are capped with a polymatroid constraint, i.e. $W(\bu) = \max\{\sum_i \overline u_i : \overline \bu \in P,~ \overline \bu \leq \bu\}$ for a polymatroid $P$. The most recent of these, \cite{hathcock2024online}, uses a principal partitioning of the polymatroid to show that a $(1 -\frac{1}{e})$-competitive algorithm exists for a wide range of such settings.
\end{enumerate}

From these works, we can see that the $(1-\frac{1}{e})$-competitive algorithm from \cite{mehta2005adwords} can be broadly extended to online resource allocation problems satisfying the two properties mentioned above. However, we note that these prior works are not comprehensive, and in fact the settings of \cite{devanur2012online}, \cite{devanur2013whole}, and \cite{hathcock2024online} are largely orthogonal to each other. Moreover, each of these lines of work uses problem-specific techniques to get their results. These observations invite the natural question: 
\begin{quote}
    Is there always $(1-\frac{1}{e})$-competitive fractional algorithm for any online resource allocation problem with concave, diminishing-returns objective?
\end{quote}
We answer this question affirmatively, giving a general $(1 - \frac{1}{e})$-competitive algorithm for any such online resource allocation problem. In doing so, we generalize the corresponding results of \cite{feldman2009online, devanur2012online, devanur2013whole, hathcock2024online} and unify them under a common framework.

\subsection{Our Results}\label{sec:results}
Formally, we give a $(1-\frac{1}{e})$-competitive algorithm for the following general problem, which we call \emph{Online Concave Diminishing-Returns Allocation (OCDRA)}. This problem is defined in the abstract as to encompass as many settings as possible.

\begin{definition}\label{def:ocdra}
    In an instance of \emph{online concave diminishing-returns allocation (OCDRA)}, we have a set $[n]$ of divisible items which arrive online one by one. When item $j \in [n]$ arrives, it reveals a set $A_j$ of possible allocation options (e.g. buyers, advertisers, configurations etc.). The algorithm must irrevocably choose some combination of these options, given by a vector $(x_a)_{a \in A_j}$ where $x_a \geq 0$ denotes the amount of item $i$ allocated through option $a$, and $\sum_{a \in A_j} x_a \leq 1$. We assume the sets $\{A_j\}_{j \in [n]}$ to be disjoint.

    The objective is to the maximize $f(x)$, where $f$ is a function such that
    \begin{enumerate}
        \item $f : \Rp^{A} \to \Rp$ for $A = \bigcup_j A_j$,
        \item $f(0) = 0$ and $f$ is concave, monotone increasing.
        \item $f$ is upward-differentiable\footnote{All multivariate functions $f : \Rp^m \to \Rp$ we consider will be assumed to be \emph{upward-differentiable}, meaning that for every $\bx \in \Rp^m$ there exists an \emph{upward-gradient} $\nabla f(\bx) \in \R^m$ such that $\lim_{\epsilon \to 0^+} (f(\bx + \epsilon \by) - f(\bx))/\epsilon = \ip{\nabla f(\bx), \by}$ for every $\by \in \Rp^m$.} and $\nabla f$ is monotone decreasing coordinate-wise (diminishing returns).
    \end{enumerate}
    We also assume that upon the $j$th arrival, the algorithm only has knowledge of $f$ restricted\footnote{Formally, we define the function $f$ restricted to $S \subseteq A$ by $x \mapsto f(x, \mathbf{0}_{A \setminus S})$ for $x \in \Rp^S$, i.e., we zero-out all inputs to $f$ at coordinates in $A \setminus S$. In other words, we assume the algorithm has no information on the dependence of $f$ on coordinates other than those it has seen.} to $A_1 \cup \dots \cup A_j$. 
\end{definition}

\begin{theorem}\label{thm:main}
    There exists a $(1-\frac{1}{e})$-competitive algorithm for online concave diminishing-returns allocation.
\end{theorem}

We recover the traditional online resource allocation model from this setting when $A = [m] \times [n]$ and $A_j = \{\{i,j\} : i \in [m]\}$, i.e. $A_j$ contains the set of edges from item $j$ to any offline agent $i \in [n]$. However, allowing arbitrary $A_j$ allows us to easily generalize the settings of \cite{devanur2013whole} (where $A_j$ represents the set of configurations of page $j$) and \cite{hathcock2024online} (where $A_j$ represents the part $Q_j$ arriving at step $j$).

Moreover, due to the generality of the setting we consider, \Cref{thm:main} also implies new results beyond those covered by prior work. For instance, we obtain results for the following settings.

\begin{itemize}
    \item \textbf{Combinations of Prior Settings.} A major strength of our approach is that, by studying such a broad class of objective functions, we can prove results for settings which combine the lines of work above. For instance, suppose we have an online resource allocation problem in which agents each have concave valuation functions $v_i(\bx_i) = M_i(\sum_j b_{ij} x_{ij})$ as in \cite{devanur2012online}, but we also have a polymatroid cap on the welfare given by $W(\bu) = \max\{\sum_i \overline u_i : \overline \bu \in P,~ \overline \bu \leq \bu\}$ as in \cite{hathcock2024online}. Although previously these two works were incompatible, \Cref{thm:main} now implies a $1 - \frac{1}{e}$ competitive ratio in the joint setting. This comes from the observation that the set of functions $f$ satisfying the conditions of \Cref{def:ocdra} are closed under operations including positive linear combinations and composition, which we show in \Cref{lem:cdr-props}.
    
    \item \textbf{Beyond $p$-mean Online Welfare Maximization.} In the online $p$-mean welfare problem introduced by \cite{barman2022universal}, we consider an online resource allocation problem in which agents have linear valuations $u_i = v_i(\bx_i) = \sum_j b_{ij} x_{ij}$, and the objective is to maximize the $p$-mean welfare $W(\bu) = (\frac{1}{n}\sum_i u_i^p)^{1/p}$, where $p \in [-\infty,~1)$. Here, the $p$-mean objective is used to capture a notion of fairness among agent utilities. Recently, the work of \cite{huang2025long} settled the optimal competitive ratio of online $p$-mean welfare maximization for all $p \leq 1$ up to lower-order terms.

    However, a separate line of work has initiated the study of ``beyond $\ell_p$'' objectives for optimization problems (e.g. \cite{azar2016online, kesselheim2023online, kesselheim2024supermodular}). In these works, optimization problems which historically have been examined with an $\ell_p$ norm objective for $p \geq 1$ are instead considered with an arbitrary convex objectives, which can capture more complex notions of fairness. Following the same paradigm in the concave $p\leq 1$ regime, we may consider online welfare maximization with general concave $W(\bu)$. Our results imply that, when $W(\bu)$ has monotone gradients (a common assumption in beyond $\ell_p$ analysis, as in \cite{azar2016online}), we obtain a $(1 - \frac{1}{e})$ competitive algorithm. This holds even when agents' valuations $v_i$ are not just linear, but are any valuation which is monotone, concave, and has diminishing returns.

    \item \textbf{Adwords with Convex Budget Constraints.} The works of \cite{wang2016matroid,zhang2020combinatorial,hathcock2024online} examined, among other things, the Adwords problem with more general budget constraints, i.e. valuation functions of the form $v_i(\bx_i) := \max\{\sum_j b_{ij}z_{ij} : \bz \leq \bx,~ (b_{ij} z_{ij})_{j \in [m]} \in K_i\}$, where $K_i \subseteq \Rp^m$ is convex, downward-closed\footnote{We say $K \subseteq \Rp^m$ is \emph{downward-closed} if for any $x \in K$ and $z \leq x$, we have $z \in K$.} set. In other words, the buyer receives linear rewards given by values $b_{ij}$, but is limited to a ``feasible region'' defined by $K_i$. When $K_i = \{\bx \in \Rp^m : \sum_j x_j \leq B_j\}$, we recover the setting of \cite{mehta2005adwords}, but more complex sets $K_i$ can capture more complex constraints, such as constraints due to network traffic or tier budgets \cite{hathcock2024online}.
    
    The aforementioned works imply $(1 - \frac{1}{e})$-competitive algorithms in the case where each $K_i$ is a polymatroid. In comparison, our work implies a $(1 - \frac{1}{e})$-competitive algorithm whenever the above valuations $v_i$ have diminishing returns. This is true when $K_i$ is a polymatroid, but more generally, when the norm $\|\cdot\|_{K_i}$ defined by $\|\by\|_{K_i} = \sup_{\alpha \in K_i} \ip{\alpha, \by}$ is a ``submodular norm,'' as in \cite{patton2023submodular}. This is a broad class which captures many common norms, $\ell_p$-norms, Top-$k$ norms, and Lov\'asz extensions of submodular functions \cite{bach2019submodular}.
\end{itemize}

\subsection{Techniques and Contributions}
To obtain our results, we use a greedy algorithm with respect to an auxiliary value function $U(x)$, along with a primal-dual analysis to bound the competitive ratio. This approach alone is not new, as the online primal-dual method is featured quite often in the study of online resource allocation. Instead, our main new contributions are as follows: First, we identify a key sufficient condition on the function $U(x)$ for our algorithm to obtain a given competitive ratio $\gamma \in [0,1]$ for our OCDRA setting. Second, we show that for $\gamma = 1-\frac{1}{e}$, a function satisfying this condition for \textit{any} problem instance is given by
\begin{equation}\label{eq:u-def}
U(\bx) := \frac{1}{e-1}\int_0^1 e^t \cdot t f(\bx/t) dt.
\end{equation}
We stress that designing online resource allocation algorithms via the above auxiliary value function is a novel perspective that has not appeared in prior work, and it is this new perspective which allows us to obtain our general and unified results.

\paragraph{Algorithm} Our algorithm for OCDRA can be formally described by the following ``continuous greedy'' algorithm with the function $U$ given in \eqref{eq:u-def}. Note that the expression of $U$ in \eqref{eq:u-def} allows us to compute $\nabla U(\bx)$ with only the partial knowledge of $f$ restricted to $A_1 \cup \dots \cup A_j$, as required by the problem setting.
\begin{algorithm}[h]
    \caption{Continuous Greedy with Respect to $U(\bx)$}
    \label{alg:con-greedy}
    \begin{enumerate}
        \item 
        When each item $j$ arrives and reveals $A_j$:
        \begin{enumerate}
            \item Initialize $x_a = 0$ for all $a \in A_j$.
            \item Over time interval $t \in [0,1]$, continuously choose $a_t \in \arg \max_{a \in A_j} \frac{\partial U(\bx)}{\partial x_a}$ and increase $x_{a_t}$ by $dt$.
        \end{enumerate}
    \item Return $\bx$.
    \end{enumerate}
\end{algorithm}

To get intuition for the formula \eqref{eq:u-def}, we can compare our algorithm to those used by prior works. Recall that a simple continuous greedy algorithm (i.e. \Cref{alg:con-greedy} with $U(\bx) = f(\bx)$) is $\frac{1}{2}$-competitive. To do better than $\frac{1}{2}$, a common theme among prior works is that the algorithm should consider not only the current marginal reward, but potential future rewards as well. In the setting of \cite{mehta2005adwords}, this idea manifests in the algorithm giving more weight to agents with greater remaining budgets, i.e. preferring to allocate to agents with larger potential for rewards in the future. 

Here, we implement this idea by considering not just the immediate marginal rewards described by $\nabla f(\bx)$, but also possible ``future'' rewards given by $\nabla f(\by)$ for values $\by > \bx$. Notice that $e-1 = \int_0^1 e^t dt$, and $\nabla U(\bx) = \frac{1}{e-1}\int_0^1 e^t \cdot \nabla f(\bx/t) dt$. In other words, the gradient of $U$ is a convex combination of the gradient of $f$ at points $\by = \bx/t$ for $t \in (0,1]$. By continuously increasing $\bx$ in coordinates which maximize $\nabla U$, the algorithm intuitively seeks to maximize a mixture of both the current reward and potential future rewards. The weighting of this convex combination is carefully chosen to obtain the optimal competitive ratio through our analysis.

\paragraph{Analysis} To analyze the competitive ratio of the algorithm, we will write a primal and dual program for an instance of OCDRA using Fenchel duality (we refer the reader to \cite{boyd2004convex} for additional background on Fenchel duality, although it is not required for our discussion). We show that if $U$ has a certain ``$\gamma$-balanced'' growth property for some $\gamma \in [0,1]$ (formally defined in \Cref{def:good}), then we can update our dual variables continuously online in a way that ensures the primal objective is always at least a $\gamma$ fraction of dual objective. Using weak duality and feasibility of our primal and dual solutions, we obtain the that competitive ratio of the algorithm is at least $\gamma$. 

The challenge in this approach lies in finding a function $U$ which is $\gamma$-balanced. We note that in the setting of \cite{devanur2012online}, the authors are able to write a differential equation in order to find a function satisfying a similar property for 1-dimensional functions. However, as our objectives are high-dimensional concave functions, it is unclear if one can generalize this approach. Our key insight is that we instead observe that the expression of $U$ in terms of $f$ given in \eqref{eq:u-def} already obtains the desired property for any $f$ with $\gamma = 1 - \frac{1}{e}$. Moreover, we know this is the best possible for general $f$ due to hard instances of Adwords \cite{mehta2005adwords}. We leave as an open question if one can obtain $\gamma > 1 - \frac{1}{e}$ for special cases of $f$.

\subsection{Further Related Work}

\paragraph{Resource Allocation with Indivisible Items} For our settings, we consider divisible items, i.e. we allow $\bx$ to be fractional. However, online resource allocation problems with \emph{indivisible items} (i.e. integral $\bx$) have also been extensively studied. In particular, we discuss work on settings without a ``small bids'' assumption as seen in \cite{mehta2005adwords, devanur2013whole, hathcock2024online}, which informally allows one to argue that the indivisible item setting is ``close'' to the setting with divisible items.

Work on such online resource allocation problems began with the foundational work of Karp, Vazirani, and Vazirani \cite{karp1990optimal} on the problem of online bipartite matching. Here, the authors prove an optimal $(1 - \frac{1}{e})$-competitive ratio using the RANKING algorithm. This algorithm was also later extended to obtain the optimal $1 - \frac{1}{e}$ competitive ratio for online vertex-weighted matching \cite{aggarwal2011online}, online bipartite $b$-matching \cite{albers2021optimal}, and online submodular welfare with matroid-rank valuation function \cite{hathcock2024online}.

However, reaching a $1 - \frac{1}{e}$ competitive ratio for more general settings has proven to be a major challenge. For instance, one of the most general forms of online resource allocation with indivisible items and diminishing-returns objectives is \emph{online submodular welfare maximization (OSWM)}, in which we seek to maximize the sum of utilities of offline agents $i \in [n]$ who each have a monotone valuation function $v_i : 2^{[n]} \to \Rp$ which is \emph{submodular}, i.e. $v_i(A) + v_i(B) \geq v_i(A \cup B) + v_i(A \cap B)$. In general, it is known that no polynomial-time algorithm can obtain competitive ratio better than $\frac{1}{2}$ for OSWM, which is obtained by the greedy algorithm \cite{kapralov2013online}.

However, several works have shown that one can beat $\frac{1}{2}$ in special cases of OSWM, including Adwords without small bids \cite{huang2020adwords} and edge-weighted bipartite matching \cite{fahrbach2020edge, gao2022improved, blanc2022multiway}. Moreover, these competitive bounds can be further improved by assuming random order arrivals \cite{korula2015online, buchbinder2019online, huang2019online} or that arrivals come from known distributions \cite{feldman2009onlineB, jaillet2014online, huang2022power}.

\paragraph{Continuously Submodular Functions}
We note that the notion of ``diminishing-returns'' (i.e. monotone decreasing gradients) that we use has been studied previously for various optimization problems, often called \emph{DR-submodularity} \cite{bian2017guaranteed, bian2017continuous}. Additionally, a weaker notion of submodularity for continuous functions has also seen much study, often simply called \emph{continuous submodularity} \cite{bach2010structured, bach2019submodular}. The difference between these properites is easiest to see for twice differentiable functions $f$. We say $f$ is DR-submodular if for any $\bx$ the Hessian has $H_f(\bx)$ has all non-positive entries, but for $f$ to be continuously submodular, we only require off-diagonal entries of $H_f(\bx)$ to be non-positive. Both notions of submodularity for continuous functions have been studied in the context of combinatorial optimization \cite{bian2017continuous, niazadeh2020optimal, patton2023submodular} and online learning \cite{chen2018online, chen2018projection, zhang2019online, sadeghi2020online}.

%% file: problem-definition.tex
\section{Properties of our Model}
Before we prove our main result, we will establish some properties of the class of objective functions $f$ we consider in \Cref{def:ocdra}. We will call such functions \emph{CDR-valuation functions}. With these properties, we may see that OCDRA does indeed capture the allocation problems of \cite{devanur2012online}, \cite{devanur2013whole}, and \cite{hathcock2024online}, as well as the extensions mentioned in \Cref{sec:results}. 

\begin{definition}\label{def:cdr}
    A function $f : \Rp^A \to \Rp$ is a \emph{concave diminishing-returns valuation (CDR-valuation)} if it satisfies the conditions in \Cref{def:ocdra}. That is, $f(0) = 0$, $f$ is monotone increasing, and $f$ satisfies 
    \begin{enumerate}
        \item $f$ is concave,
        \item $\nabla f(\bx)$ is monotone decreasing coordinate-wise in $\bx$.
    \end{enumerate}
\end{definition}
We remark that neither properties (1) or (2) in \Cref{def:cdr} are implied by the other, despite their similarities at first glance. To see this, notice that if $f$ is twice differentiable, then property (1) is equivalent to the Hessian of $f$ being negative semi-definite everywhere, whereas property (2) is equivalent to the Hessian to having all non-positive entries.

To better understand \Cref{def:cdr}, we observe some examples of CDR-valuations.
\begin{lemma}[CDR-Valuation Examples]\label{lem:cdr-examples}
The following functions $f : \Rp^m \to \Rp$ are examples of CDR-valuations.
\begin{enumerate}
    \item \textbf{Linear} $f(\bx) = \sum_i b_{i} x_{i}$ where $b_i \geq 0$ for $i \in [m]$.
    \item \textbf{Budget-Additive} $f(\bx) := \min\{\sum_i b_i x_i,~B\}$, where $b \in \Rp^m$ and $B \geq 0$. 
    \item \textbf{Concave-of-Linear} $f(\bx) := M(\sum_i b_i x_i)$ for some $b \in \Rp^m$ and concave, non-decreasing $M : \Rp \to \Rp$ with $M(0) = 0$.
\item \textbf{Polymatroid Budget-Additive} $f(x) := \max\{\sum_i z_i : \bz \in \Rp^m,~\bz \leq \bx,~\forall S \subseteq [m]~ \sum_{i \in S} z_i \leq r(S)\}$, where $r : 2^{[m]} \to \Rp$ is a monotone submodular function with $r(0) = 0$. In other words, $f(\bx)$ is the maximum value of a point $\bz \leq \bx$ and in the polymatroid with rank function $r$.
\end{enumerate}
\end{lemma}
\begin{proof}
    We note that examples 1 and 2 are special cases of 3, so we need only show that 3 and 4 are CDR-valuations.

    For (3), we it is easy to see that $f(0) = M(0) = 0$. Additionally, we have $\frac{\partial}{\partial x_i} f(\bx) = M'(\sum_i b_i x_i) \cdot b_i$. This is positive and decreasing in $\bx$ by concavity of $M$, so we have that $f$ is monotone increasing and $\nabla f$ is monotone decreasing coordinate-wise. Finally, we observe that $f$ is concave as 
    \begin{align*}
        f(\lambda \bx + (1 - \lambda) \by) &= M\left(\sum_i(\lambda b_ix_i + (1 - \lambda) b_iy_i)\right) \\
        &\geq \lambda M\left(\sum_i b_ix_i\right) + (1 - \lambda)M\left(\sum_i b_iy_i\right) = \lambda f(\bx) + (1 - \lambda )f(\by).
    \end{align*}

    For (4), we first observe that the function $f$ is monotone, has $f(0) = 0$. To see that $f$ is concave, let $\bx, \by \in \Rp^m$, and suppose $\bz_\bx$ and $\bz_\by$ achieve the maximum in the definition for $f(\bx)$ and $f(\by)$ respectively. Then for $\lambda \in [0,1]$, the point $\bz := \lambda \bz_\bx + (1 - \lambda) \bz_\by$ is feasible in the maximum for $f(\lambda \bx + (1-\lambda) \by)$. Hence, $f(\lambda \bx + (1-\lambda) \by) \geq \sum_i z_i = \lambda f(\bx) + (1 - \lambda) f(\by)$.
    
    To show that $\nabla f$ is monotone decreasing, we need to use the ``submodular water-levels'' machinery established in \cite{hathcock2024online}. We claim $\frac{\partial f(\bx)}{\partial x_i} = \ind\{{w_i\up \bx < 1}\}$, where $\bb w\up x \in \Rp^m$ is the water-level vector defined in Definition 3.1 of \cite{hathcock2024online}. Using the Proposition 3.4 in \cite{hathcock2024online}, we know that $w_i\up \bx$ is monotone increasing in $\bx$, so our claim implies that $\frac{\partial f(\bx)}{\partial x_i}$ is monotone decreasing, as desired. 

    Showing this claim is not too difficult, but it requires more technical properties of submodular water-levels. As this is tangential to our main results, we defer the remainder of the proof to \Cref{apx:wl-claim}.
\end{proof}

\subsection{Operations on CDR Valuations}

In addition, we note the following closure properties of the class of CDR-valuations, which allow us to combine and modify CDR-valuations to generate new ones.

\begin{lemma}[Operations Preserving CDR-Valuations]
The class of CDR-valuations is closed under the following operations.
    \begin{enumerate}
        \item \textbf{Positive Linear Combinations.} If $f_1, \dots, f_k$ are CDR-valuations on $\Rp^m$, and $\lambda_1, \dots, \lambda_k \geq 0$, then $f = \sum_{i=1}^k \lambda_i f_i$ is a CDR-valuation.
        \item \textbf{Positive Linear Transformation of Inputs.} Suppose $f : \Rp^k \to \Rp$ is a CDR-valuation, and $A \in \Rp^{k \times m}$ is a matrix with non-negative entries. Then $h : \Rp^{m} \to \Rp$ given by $h(\bx) = f(A \bx)$ is a CDR-valuation.
        \item \textbf{Composition.} Suppose $g_1, \dots, g_k$ are CDR-valuations such that $g_i:\Rp^{m} \to \Rp$, and $f : \Rp^k \to \Rp$ is a CDR-valuation. Define the function $h : \Rp^{m} \to \Rp$ by
        \[
        h(\bx) = f\Big(g_1(\bx),~\dots~, g_k(\bx)\Big).
        \]
        Then $h$ is a CDR-valuation.
    \end{enumerate}
\end{lemma}
\begin{proof}
    For (1), we simply observe that each property of \Cref{def:cdr} is preserved under positive linear combinations. If $f_1, \dots, f_k$ are CDR-valuations, then $f = \sum_{i \in [k]} \lambda_i f_i$ is concave, monotone increasing, and has $f(0) = \sum_{i \in [k]} \lambda_i f(0) = 0$. Moreover, since $\nabla f = \sum_{i \in [k]} \lambda_i \nabla f$, it is easy to see $\nabla f$ is monotone decreasing since all $\nabla f_i$ are monotone decreasing.

    Next, we have that (2) is a special case of (3), where we take $g_i(\bx) = \sum_{j \in [m]} A_{ij} x_j$. Hence, it only remains to show property 3. First, we have $h(0) = f(g_1(0), \dots, g_k(0)) = f(0) = 0$. Next, computing partial derivatives of $h$ gives
    \[
    \frac{\partial h(\bx)}{\partial x_j} = \sum_{i \in [k]} \frac{\partial g_i (\bx)}{\partial x_j} \cdot \frac{\partial f (g_1(\bx), \dots, g_k(\bx))}{\partial (g_i(\bx))}.
    \]
    Notice that this expression is non-negative, as all terms are non-negative. Thus, $h$ is monotone increasing. Additionally, since all $g_i$ are monotone increasing, and $\nabla f$ is monotone decreasing, we have that all terms are monotone decreasing in $\bx$. Therefore, $\nabla h(\bx)$ is decreasing in $\bx$. 

    Lastly, it remains to check that $h$ is concave. For $\lambda \in [0,1]$ and $\bx, \by \in \Rp^m$, we have
    \begin{align*}
        h(\lambda \bx + (1 - \lambda)\by) &= f\Big(\big(g_i(\lambda \bx + (1 - \lambda) \by)\big)_{i \in [k]}\Big) \\
        &\geq  f\Big(\big(\lambda g_i(\bx) + (1 - \lambda) g_i(\by)\big)_{i \in [k]}\Big)&&\text{by concavity of $g_i$ and monotonicity of $f$,}\\
        & \geq \lambda f\Big(\big(g_i(\bx)\big)_{i \in [k]} \Big) + (1 - \lambda) f\Big(\big(g_i(\by)\big)_{i \in [k]}\Big)&&\text{by concavity of $f$,}\\
        &=\lambda h(\bx) + (1 - \lambda) h(\by).
    \end{align*}
\end{proof}

\subsection{Capturing Prior Settings} Given \Cref{lem:cdr-examples} and \Cref{lem:cdr-props}, we can now see how OCDRA captures the settings of \cite{devanur2012online}, \cite{devanur2013whole}, and \cite{hathcock2024online}.

\paragraph{Online Matching with Concave Returns \cite{devanur2012online}} In this setting, we have $A = [m] \times [n]$, and each $A_j = \{\{i,j\} : i \in [m]\}$. Our objective $f$ has the form
    \[
    f(\bx) = \sum_{i \in [m]} M_i(\sum_{j \in [n]}b_{ij}x_{ij}),
    \]
    where $b_{ij} \geq 0$ and each $M_i : \Rp \to \Rp$ is concave, monotone increasing, and has $M_i(0) = 0$. From \Cref{lem:cdr-examples}, each function $M_i(\sum_{j \in [n]}b_{ij}x_{ij})$ is a CDR-valuation, so property (1) of \Cref{lem:cdr-props} tells us that $f$ is a CDR valuation.
\paragraph{Online Whole Page Optimization \cite{devanur2013whole}} 
    In this setting, we have a set $[m]$ of offline agents with budgets $B_i$, and each arrival $j \in [n]$ has a set $A_j$ representing different ``configurations'' $a \in A_j$ in which $j$ can be allocated. Each $a \in A_j$ consumes budget $b_{i, a} \geq 0$ from agent $i$ and provides reward $w_{i, a} \geq 0$. The objective is then 
    \[
    f(\bx) = \sum_{i \in [m]} \max\left\{\sum_{a \in A}w_{i,a} z_{i,a} ~:~\sum_{a \in A} b_{i,a} z_{i,a} \leq B_i;~~ \forall a \in A,~0 \leq z_{i,a} \leq x_a\right\}.
    \]
    To see that this is a CDR-valuation, we note that we can first express $f$ as
    \[
    f(\bx) = \int_0^\infty dt \cdot \sum_{i \in [m]} \min\left\{\sum_{a ~:~ t \leq \frac{w_{i,a}}{b_{i,a}}}b_{i,a} x_{i,a},~B_i \right\}.
    \]
    This is a positive linear combination of budget additive functions, so from \Cref{lem:cdr-examples} and \Cref{lem:cdr-props} we have that $f$ is a CDR-valuation.

\paragraph{Online Submodular Assignment \cite{hathcock2024online}} 
    In this setting, we have a monotone submodular rank function $r : 2^{A} \to \Rp$ which defines a polymatroid $P_r := \{\bz \in \Rp^A : \forall S \subseteq A,~ \sum_{i \in S} z_i \leq r(S)\}$. Additionally, we have costs $b_a \geq 0$ and values $w_a \geq 0$ for each $a \in A$. Our objective is then 
    \[
    f(\bx) = \max \left\{ \sum_{a \in A} w_a z_a ~:~ \forall S \subseteq A,~\sum_{a \in S} b_a z_a \leq r(S);~~\forall a \in A,~0\leq z_a\leq x_a\right\}.
    \]
    Again, we can decompose this objective as a integral over ``bang-per-buck'' levels $\frac{w_a}{b_a}$ to get
    \[
    f(\bx) = \int_0^\infty dt \cdot \max \left\{ \sum_{a : \frac{w_a}{b_a} \geq t} b_a z_a ~:~ \forall S \subseteq A,~\sum_{a \in S : \frac{w_a}{b_a} \geq t} b_a z_a \leq r(S);~~\forall a \in A,~0\leq z_a\leq x_a\right\}.
    \]
    From \Cref{lem:cdr-examples}, and \Cref{lem:cdr-props}, we know that this inner maximum is a CDR-valuation, since it is a polymatroid budget-additive function over the linear transformation of $\bx$ given by $(b_a x_a)_{a \in A}$. Since $f$ is a positive linear combination of such functions, we have that $f$ is a CDR-valuation.

%% file: main-proof.tex
\section{Primal-Dual Proof of \Cref{thm:main}}
We now prove our main theorem using an online primal-dual approach. To do so, we first establish a notion of duality based on Fenchel duality for convex functions. However, as we are working with positive concave functions, it will be convenient for us to use the following notation.
\begin{definition}\label{def:dual}
    For a concave function $f : \Rp^A \to \Rp$, we will use $\hat f : \Rp^A \to \Rp$ to denote the function given by
    \[
    \hat f(\alpha) := \sup_{\bx \in \Rp^A} \left(f(\bx) - \ip{\alpha, \bx}\right).
    \]
    Note that $\hat f(\alpha) = (-f)^*(-\alpha)$, where $(-f)^*$ denotes the Fenchel dual of the convex function $-f$.
\end{definition}

Intuitively, for a function $f$ in 1 dimension, $\hat f(\alpha)$ gives the y-intercept of the tangent line to $f$ with slope $\alpha$. In general, $\hat f$ gives the constant term in the linear approximation to $f$ with linear component $\ip{\alpha, \cdot}$. We note that $\hat f$ has the following properties, due to its relation to the Fenchel dual.
\begin{proposition}\label{lem:cdr-props}
    Let $f$ be a CDR-valuation function. Then we have
    \begin{enumerate}[(1)]
        \item $\hat f$ is non-negative, convex, and monotone decreasing on $\Rp^A$.
        \item For all $\bx \in \Rp^A$, we have $\hat f(\nabla f(\bx)) = f(\bx) - \ip{\nabla f(\bx),~\bx}$.
    \end{enumerate}
\end{proposition}
\begin{proof}
    To prove property (1), notice that $\hat f$ is a supremum over decreasing affine functions. Hence, $\hat f$ is monotone decreasing and convex. To see that $\hat f$ is non-negative, notice that taking $\bx = 0$ gives $\hat f(\alpha) \geq f(0) - \ip{a, 0} = 0$.

    To see property (2), notice that for any $\bx, \by \in \Rp^A$, we have $f(\bx) + \ip{\nabla f(\bx),~(\by-\bx)} \geq f(\by)$ by concavity of $f$. Rearranging gives
    \[
    f(\bx) - \ip{\nabla f(\bx),~\bx} \geq f(\by) - \ip{\nabla f(\bx), \by}.
    \]
    Hence, we conclude $f(\bx) - \ip{\nabla f(\bx),~\bx} = \sup_{\by \in \Rp^A} \left(f(\by) - \ip{\nabla f(\bx), \by}\right) = \hat f(\nabla f(\bx))$.
\end{proof}

\paragraph{Primal and Dual Programs} Using this notation, we can write primal and dual programs for OCDRA. Note that for a given instance of OCDRA, we can represent the problem by the concave program below.
\begin{align*}
    \max~&~f(\bx), \\
    \text{s.t.} ~&~ \sum_{a \in A_j} x_a \leq 1 \quad \forall j \in [n],\\
    &~\bx \geq 0.
\end{align*}
Additionally, using the function $\hat f$ from \Cref{def:dual}, we can write a dual convex program as
\begin{align}
    \min~&~\hat f(\alpha) + \sum_j \beta_j, \nonumber\\
    \text{s.t.} ~&~ \beta_j \geq \alpha_a \quad \forall j \in [n],~a \in A_j, \label{eq:dual-ineq}\\
    &~\alpha, \beta \geq 0.\nonumber
\end{align}
Although it may not be clear immediately, we can verify that this program indeed satisfies a weak duality property with our primal. This allows us to use the value of any feasible dual solution as an upper bound on the optimal primal objective.
\begin{lemma}[Weak Duality]
    Let $f:\Rp^A \to \Rp$ be a concave function, and suppose $\bx$ and $(\bb \alpha, \bb \beta)$ are feasible solutions to the above primal and dual programs respectively. Then $f(\bx) \leq \hat f(\alpha) + \sum_j \beta_j$.
\end{lemma}
\begin{proof}
    We have
    \begin{align*}
        \hat f(\bb \alpha) + \sum_j \beta_j 
        &\geq f(\bx) - \ip{\bb \alpha, \bx} + \sum_j \beta_j&&\text{by definition of $\hat f$,}\\
        &\geq f(\bx) - \sum_j\sum_{a \in A_j} \alpha_a x_a + \sum_j \beta_j \sum_{a \in A_j} x_a &&\text{since $\sum_{a \in A_j}x_a \leq 1$,}\\
        &= f(\bx) + \sum_j\sum_{a \in A_j} (\beta_j - \alpha_a) x_a\\
        &\geq f(\bx)&&\text{since $\beta_j \geq \alpha_a$ for $a \in A_j$}.
    \end{align*}
\end{proof}

\paragraph{A Sufficient Condition for $\gamma$-Competitiveness} Using this dual program, we can now define the property of $U$ which will allow us to bound the competitive ratio of \Cref{alg:con-greedy}.

\begin{definition}\label{def:good}
    For a CDR-valuation function $f$ and $\gamma \in [0,1]$, we say a function $U$ is \emph{$\gamma$-balanced} with respect to $f$ if $U$ is a CDR-valuation function such that for any $\bx \in \Rp^A$, we have
    \begin{equation}\label{eq:balanced}
    \frac{1}{\gamma}f(\bx) \geq U(\bx) + \hat f\left(\nabla U(\bx)\right).
    \end{equation}
\end{definition}
We call such a function ``balanced,'' since $U$ must balance the contribution of both terms on the RHS of \eqref{eq:balanced}. If $U$ grows too quickly, then $U(\bx)$ will be large, but if $U$ grows too slowly, then $\nabla U(\bx)$ will be small, and hence $\hat f(\nabla  U(\bx))$ will be large. Thus, the ideal $U$ for a function $f$ will be one which grows not too quickly and not too slowly, so that \eqref{eq:balanced} holds for the largest possible value of $\gamma$.

\begin{theorem}
    Suppose $U$ is $\gamma$-balanced with respect to  $f$. Then for an instance of OCDRA with objective $f$, the continuous greedy algorithm with respect to $U$ given by \Cref{alg:con-greedy} is $\gamma$-competitive.
\end{theorem}

\begin{proof}
    Over the course of the algorithm, we will update dual variables $\alpha$ and $\beta$ continuously along with $\bx$. To track our primal and dual variable updates, let $\alpha \up {j,t}$, $\beta \up {j,t}$, and $\bx \up {j,t}$ denote the values of $\alpha$, $\beta$, and $\bx$ respectively during the $j$th arrival at each time $t \in [0,1]$. Recall that at time $t$, we increase $x_{a_t}$ by $dt$ for some $a_t \in \arg \max_{a \in A_j} \frac{\partial}{\partial x_a} U(\bx\up {j,t})$. As we increase $x$ at time $t$, we also increase $\beta_j$ by $d\beta_j := \frac{\partial}{\partial x_{a_t}} U(x\up {j,t}) dt$, and maintain $\alpha \up {j,t} = \nabla U(x \up {j,t})$ at all times.
    
    Notice that at time $t = 1$, we must have $\beta_j \up {j,1} \geq \max_{a \in A_j} \alpha_a \up {j,1}$, since
    \[
    \beta_j \up {j,1} = \int_0^1 \frac{\partial}{\partial x_{a_t}} U(\bx \up {j,t})dt =\int_0^1 \max_{a \in A_j} \frac{\partial}{\partial x_{a}} U(\bx \up {j,t})dt \geq \max_{a \in A_j} \frac{\partial}{\partial x_{a}} U(\bx \up {j,1})= \max_{a \in A_j} \alpha_a \up {j,1},
    \]
    where the inequality follows from the fact that $\nabla U(\bx\up t)$ is decreasing over time.
    This implies that $\alpha \up{j,1}$ and $\beta \up {j,1}$ satisfy \eqref{eq:dual-ineq}. Moreover, for the remaining course of the algorithm, $\alpha_a$ is monotonically decreasing and $\beta_j$ unchanged. Hence, our final dual values for $\alpha$ and $\beta$ at the end of the algorithm also satisfy \eqref{eq:dual-ineq}. Since this holds for all $j$, and our dual variables are non-negative, we have that the final $\alpha$ and $\beta$ values are feasible.

    Next, we compare the primal and dual objectives at the end of the algorithm. We hencefore will use $\bx$, $\alpha$, and $\beta$ to the denote the final values of the primal and dual variables when the algorithm completes. 
    
    Notice first that $\beta_j$ is exactly equal to the change in the value of $U(\bx)$ upon the $j$th arrival, as 
    \[\beta_j = \int_0^1 \frac{\partial}{\partial x_{a_t}} U(\bx \up {j,t})dt = \int_0^1 \ip{\nabla U(\bx \up {j,t}),~\frac{d \bx\up{j,t}}{dt}} dt = \int_0^1 \frac{d U(\bx\up{j,t})}{dt} \cdot dt = U(\bx\up{j,1}) - U(\bx\up{j,0}),
    \]
    and so $\sum_j \beta_j = U(\bx) - U(0) =  U(\bx)$. Hence, since $U$ is $\gamma$-balanced, we have
    \[
    f(\bx) \geq \gamma\left(U(\bx) + \hat f(\nabla U(\bx))\right) = \gamma (\sum_j \beta_j + \hat f(\alpha)).
    \]
    By weak duality, we have that $f(\bx) \geq \gamma f(\bx^*)$ for any feasible primal solution $\bx^*$, which completes our proof.
\end{proof}

\subsection{Existence of a Balanced Function}
To complete our proof of \Cref{thm:main}, it only remains to show that we can find a function $U$ which is $(1 - \frac{1}{e})$-balanced with respect to any CDR-valuation $f$. We prove the function given in \eqref{eq:u-def} satisfies this property.
\begin{lemma}\label{lem:DR-decomp}
    For any CDR-valuation function $f : \Rp^A \to \Rp$, define the function $U : \Rp^A \to \Rp$ by
    \[
    U(\bx) = \frac{1}{e-1}\int_0^1 e^t \cdot t f(\bx/t) dt.
    \]
    Then $U$ is $(1 - \frac{1}{e})$-balanced with respect to $f$.
\end{lemma}
\begin{proof}
    We check the conditions of \Cref{def:good}. First, we clearly have $U(0) = 0$. Next, notice that $U(\bx)$ is a convex combination of functions of the form $t \cdot f(\bx/t)$. Each of these functions is a CDR-valuation, as it is a scaling of the function $f$. Since these properties are preserved under positive linear combinations, we have that $U$ is also a CDR-valuation.

    To verify the second property, we compute
    \begin{align*}
        \hat f(\nabla U(\bx)) &= \hat f\left(\frac{1}{e-1}\int_0^1 e^t \cdot \nabla f(\bx/t)dt\right)\\
        &\leq \frac{1}{e-1}\int_0^1 e^t \cdot \hat f(\nabla f(\bx/t))dt&&\text{by convexity of $\hat f$,}\\
        &=\frac{1}{e-1}\int_0^1 e^t \cdot \left(f(\bx/t) - \ip{\nabla f(\bx/t),~\bx/t}\right)dt&&\text{by property (2) of \Cref{lem:cdr-props}}
    \end{align*}
    Hence, we have
    \[
    U(\bx) + \hat f(\nabla U(\bx)) \leq \frac{1}{e-1}\int_0^1 e^t \cdot \Big((1+t)f(\bx/t) - \ip{\nabla f(\bx/t),~\bx/t}\Big)dt
    \]
    We seek to show that the RHS expression is at most $\frac{e}{e-1} \cdot f(x)$. To do so, consider the the function $g : \Rp \to \Rp$ given by $g(u) = f(u \bx)$. This means $g'(u) = \ip{\nabla f(u \bx),~\bx}$. Now the inequality we seek can be expressed in terms of $g$ as
    \[
    \int_0^1 e^t \left( (1 + t)g(\frac{1}{t}) - \frac{1}{t} g'(\frac{1}{t})\right) dt \leq e \cdot g(1).
    \]
    Verifying this inequality is simply a matter of calculus. By substituting $u = \frac{1}{t}$, we have
    \begin{align*}
        \int_0^1 e^t \left( (1 + t)g(\frac{1}{t}) - \frac{1}{t} g'(\frac{1}{t})\right) dt 
        &= \int_\infty^1 \left(-\frac{e^{1/u}}{u^2}(1+\frac{1}{u})g(u) + \frac{e^{1/u}}{u}g'(u)\right)du\\
        &=\int_\infty^1 \left(\frac{e^{1/u}}{u}g(u)\right)'du\\
        &= e \cdot g(1) - \lim_{u \to \infty} \frac{g(u)}{u} \leq e \cdot g(1).
    \end{align*}
\end{proof}

%% file: appendix.tex
\section{Missing Proof for \Cref{lem:cdr-examples}}
\label{apx:wl-claim}
Here, we prove the claim that $\frac{\partial f(\bx)}{\partial x_i} = \ind\{w_i\up \bx < 1\}$, where $f$ is defined in bullet (4) of \Cref{lem:cdr-examples}.

To show our claim, it suffices to show that $f(\bx) = L_r((\min\{1,~w_i\up \bx\})_{i \in [m]})$, where $L_r$ is the Lov\'asz extension of the submodular function $r$, i.e.
\[
L_r(\bb w) = \int_0^\infty r(\{i \in [m] : w_i \geq t\}) dt.
\]
This fact implies our claim due to the ``chain-rule'' Lemma 3.8 of \cite{hathcock2024online}, since if we define $G : \Rp \to \Rp$ by $G(x) = \min\{1, x\}$, we then have
\[
\frac{\partial f(\bx)}{\partial x_i} = \frac{\partial L_r((G(w_i\up \bx))_{i \in [m]})}{\partial x_i} = G'(w_i \up \bx) = \ind\{w_i\up \bx < 1\}.
\]
To show this fact, define the vector $\bz^* \in \Rp^m$ by
\[
z^*_i := \begin{cases}
    x_i & w_i\up \bx \leq 1\\
    \frac{x_i}{w_i \up \bx} & w_i\up \bx > 1.
\end{cases}
\]
It is not hard to check that $w\up{\bz^*}_i = \min\{1, w_i \up \bx\}$ for each $i$. Since $\bb w\up{\bz^*} \leq 1$, we have that $\bz^*$ is in the polymatroid defined by $r$, i.e. $\sum_{i \in S} z^*_i \leq r(S)$ for every $S \subseteq [m]$. Additionally, it is easy to see that $\bz^* \leq \bx$, so we have $f(\bx) \geq \sum_{i \in [m]} z^*_i = L_r(\bb w \up {z^*})$ by definition of $f$ and Proposition 3.6 of \cite{hathcock2024online}.

Next, notice that if $\bz \in \Rp^m$ is such that $\bz \leq \bx$ and $\sum_{i \in S} z_i \leq r(S)$ for every $S \subseteq [m]$ then we have (1) $\bb w\up \bz \leq \bb w \up \bx$ by monotonicity of $\bb w \up x$ in $\bx$, and (2) $\bb w \up \bz \leq 1$ by Proposition 3.5 of \cite{hathcock2024online}, since $\bz$ is feasible in the polymatroid given by $r$. Hence, we have $w_i \up \bz \leq \min\{1,~w_i\up \bx\}$ for each $i$, which gives $\bb w \up \bz \leq \bb w \up {\bz^*}$. Using monotonicity of $L_r$ and Proposition 3.6 of \cite{hathcock2024online} again, this implies $\sum_i z_i =  L_r(\bb w\up \bz) \leq L_r(\bb w\up {\bz^*})$.

Finally, as $f(\bx)$ is defined by the maximum over such $z$, we see that $f(\bx) \leq L_r(\bb w\up {\bz^*})$, and hence $f(\bx) = L_r(\bb w\up {\bz^*}) = L_r((\min\{1,~w_i\up \bx\})_{i \in [m]})$ as desired.